\documentclass[journal]{IEEEtran}

\IEEEoverridecommandlockouts
\usepackage{amsfonts}
\usepackage{color}
\usepackage{graphicx}
\usepackage[dvips]{epsfig}
\usepackage{graphics} 
\usepackage{times} 
\usepackage[cmex10]{amsmath} 
\usepackage{amssymb}  
\usepackage{cite}

\usepackage{multirow}

\def\twon #1{\left\|#1\right\|_2}
\def\onen #1{\left\|#1\right\|_1}

\def\sgn #1{\text{sgn}#1}
\def\abs #1{\left|#1\right|}
\def\inp #1{\left\langle#1\right\rangle}

\def\bC{\mathbb{C}}
\def\bR{\mathbb{R}}

\def\cN{\mathcal{N}}

\def\cL{\mathcal{L}}

\def\bee{\begin{equation}}
\def\ene{\end{equation}}

\def\beq{\begin{eqnarray}}
\def\enq{\end{eqnarray}}

\newtheorem{prop}{Proposition}

\newtheorem{find}{Finding}

\def\equ #1{\begin{equation}#1\end{equation}}

\def\sbra #1{\left(#1\right)}
\def\mbra #1{\left[#1\right]}
\def\lbra #1{\left\{#1\right\}}

\def\m #1{\emph{\textbf{#1}}}

\def\po{\left(\text{BP}^{o}\right)}

\title{On Phase Transition of Compressed Sensing in the Complex Domain}
\author{Zai Yang, Cishen Zhang, and Lihua Xie, \emph{Fellow, IEEE}
\vspace{-0.5cm}
\thanks{Z. Yang and L. Xie are with EXQUISITUS, Centre for E-City, School of Electrical and Electronic Engineering, Nanyang Technological University, 639798, Singapore (e-mail: yang0248@e.ntu.edu.sg; elhxie@ntu.edu.sg).

C. Zhang is with the Faculty of Engineering and Industrial Sciences, Swinburne University of Technology, Hawthorn VIC 3122, Australia (e-mail: cishenzhang@swin.edu.au).}}
\begin{document}
\maketitle
\begin{abstract}
The phase transition is a performance measure of the sparsity-undersampling tradeoff in compressed sensing (CS). This letter reports our first observation and evaluation of an empirical phase transition of the $\ell_1$ minimization approach to the complex valued CS (CVCS), which is positioned well above the known phase transition of the real valued CS in the phase plane. This result can be considered as an extension of the existing phase transition theory of the block-sparse CS (BSCS) based on the universality argument, since the CVCS problem does not meet the condition required by the phase transition theory of BSCS but its observed phase transition coincides with that of BSCS. Our result is obtained by applying the recently developed ONE-L1 algorithms to the empirical evaluation of the phase transition of CVCS.
\end{abstract}

\begin{keywords}
Compressed sensing, complex signals, $\ell_1$ minimization, ONE-L1 algorithms, Phase transition, Joint sparsity.
\end{keywords}
\section{Introduction}
Compressed sensing (CS) aims at recovering a signal from reduced number of linear measurements under a sparsity/compressibility condition. Different from the conventional linear recovery approach, a nonlinear scheme is used in CS to solve the following basis pursuit (BP) problem
\[\min_{\m{x}}\onen{\m{x}}, \text{ subject to }\m{A}\m{x}=\m{b},\]
where $\m{x}\in\bC^{N}$ is a sparse signal to be recovered, $\m{A}\in\bC^{n\times N}$ is a sensing matrix and $\m{b}\in\bC^n$ is a vector of sample data, with $N$ and $n$ being the signal length and sample size respectively and typically $n\ll N$.

The theory of CS is mainly focused on studying how aggressively a sparse signal can be undersampled while still preserving all information for its recovery. The existing results include those based on incoherence\cite{candes2007sparsity}, restricted isometry property (RIP)\cite{candes2006compressive} and phase transition theory\cite{donoho2010counting}. While the incoherence and RIP are sufficient conditions for the sparse signal recovery and typically quite conservative in practice\cite{blanchard2010compressed}, the phase transition defined for $N\rightarrow \infty$ is most precise owing to its necessary and sufficient condition for measuring the sparsity-undersampling tradeoff performance. For the real valued CS (RVCS) using the sensing matrix $\m{A}$ with i.i.d. random Gaussian entries, the sparsity-undersampling tradeoff of the BP is controlled by the sampling ratio $\delta=n/N$ and sparsity ratio $\rho=k/n$, where $k$ denotes the number of nonzero entries of $\m{x}$. As a result, the plane of $\sbra{\delta,\rho}$ is divided by a phase transition curve into two phases, a `success' phase where BP successfully recovers the sparse signal and a `failure' phase where the original signal cannot be recovered by solving BP, both with an overwhelming probability.

In RVCS, three different approaches, combinatorial geometry \cite{donoho2010counting}, null space method \cite{stojnic2009various} and state evolution \cite{donoho2009message}, to the phase transition of BP have provided identical results. Numerical simulations have shown that the observed phase transition matches the theoretical curve even for modestly large $N$, e.g., $N=1000$ \cite{yang2010orthonormal,donoho2009observed}, where the phase transition for finite-$N$ is defined as the value of $\rho$ at which the original signal is successfully recovered with the probability of $50\%$. Moreover, it is practically observed that the Gaussian condition on $\m{A}$ can be considerably relaxed, resulting in the well known observed universality of phase transitions \cite{donoho2009observed}.

To the best of our knowledge, the existing results on phase transition, before our preliminary work \cite{yang2011sparsity}, only deal with the RVCS problem. This letter presents our first study of the phase transition in the more general complex domain where both the signal and sensing matrix are complex valued. Its significance is twofold: 1) discovering a new phase transition for the complex valued CS (CVCS) so to inspire completeness of the phase transition theory; 2) providing insight into performance of numerous CS applications involving complex data, e.g., magnetic resonance imaging (MRI)\cite{lustig2007sparse}, radar imaging\cite{herman2009high} and source localization\cite{yang2011off}.

This letter shows that the recently developed orthonormal expansion $\ell_1$-minimization (ONE-L1) algorithms \cite{yang2010orthonormal} can be extended to CVCS after some modifications. The algorithms can efficiently compute the optimal (or numerically optimal) solution of BP, so to provide an effective means for empirically exploring the phase transition of CVCS. It is observed that not only a transition curve exists in the phase plane of CVCS, but also it is well above that of RVCS. Moreover, the universality of phase transitions across many different matrix ensembles also holds in CVCS.

The sparsity of a complex signal implies joint sparsity of its real and imaginary parts with the same support. This connects to the block-sparse CS (BSCS) problem which is to recover real valued signals with entries clustered into sparse blocks of equal size. The phase transition of the BSCS is analyzed in \cite{stojnic2009block} by Stojnic under the condition that the null space of the real valued random sensing matrix is distributed uniformly in the Grassmanian. It is shown in this letter that the phase transition of the CVCS coincides with that of the BSCS with the block size $2$. Analysis is further carried out to show that the phase transition of the CVCS is not a special case of that of the BSCS because the CVCS does not meet the aforementioned uniformly distributed null space condition.

Section \ref{section_ONE_L1} retrospects the main points in \cite{yang2010orthonormal} on ONE-L1 algorithms and extends the algorithms to the CVCS. Section \ref{section_pt_complex} applies the extended ONE-L1 algorithms to explore the phase transition of CVCS and its connection to the BSCS. Conclusion is drawn in Section \ref{section_conclusion}.

\section{ONE-L1 Algorithms for CVCS}\label{section_ONE_L1}

The ONE-L1 algorithms \cite{yang2010orthonormal} solve the BP in RVCS where $\m{x}$, $\m{A}$ and $\m{b}$ are all real valued. Assume that the sensing matrix $\m{A}$ is partially orthonormal, i.e., $\m{A}\m{A}'=\m{I}$ with $'$ denoting the transpose and $\m{I}$ being an identity matrix, and let $\Gamma({\m{d}})$ be an operator projecting the vector ${\m{d}}$ onto its first $n$ entries. The orthonormal expansion technique is to introduce an expanded orthonormal matrix $\Phi$ to reformulate the BP into
\[\po\quad\begin{array}{l}\min_{(\m{x},{\m{d}})}\onen{\m{x}}\\ \text{subject to} \ \Phi \m{x}={\m{d}}\text{ and }\Gamma({\m{d}})=\m{b},\end{array}\]
where the first $n$ rows of $\Phi$ compose $\m{A}$ and the rest orthonormal $N-n$ rows are arbitrary.

The augmented Lagrange multiplier (ALM) method is applied in \cite{yang2010orthonormal} to $\po$ and the obtained result leads to an exact ONE-L1 (eONE-L1) algorithm for iterative computation of the optimal solution. While eONE-L1 has an inner iteration loop embedded in the outer loop iteration, its relaxed version, rONE-L1, further speeds up the computation by simplifying the inner loop iteration into a single update. The rONE-L1 algorithm is numerically optimal in terms of the sparsity-undersampling tradeoff under reasonable parameter settings, i.e., its phase transition result matches that of BP. It has been shown that rONE-L1 is of iterative thresholding type and is very fast, with appropriate settings of the regulation variable, in comparison with other state-of-the-art algorithms.

An important operation in the ONE-L1 algorithms is the soft thresholding defined as
\[S_\epsilon (w)=\sgn\sbra{w}\cdot(\abs{w}-\epsilon)^+,\]
for $w\in\bR$, where $\epsilon\in\bR_+$, $(\cdot)^+=\max(\cdot,0)$ and
\[\sgn\sbra{w}=\left\{
\begin{array}{ll}w/\abs{w},&w\neq0;\\0,&w=0.
\end{array}\right.\]
In fact, $S_{\epsilon}\sbra{\m{w}}$ operates elementwisely on a real valued vector $\m{w}$. It has been well known that the soft thresholding solves the following $\ell_1$ regularized least squares problem
\equ{S_{\epsilon}\sbra{\m{w}}=\arg\min_{\m{v}} \lbra{\epsilon
\onen{\m{v}}+\frac{1}{2}\twon{\m{w}-\m{v}}^2}\label{formula_soft_threshold}}
for real valued vectors $\m{w},\m{v}$ of the same dimension.

To extend the ONE-L1 algorithms to CVCS, the complex sensing matrix $\m{A}\in\bC^{n\times N}$ is also assumed to satisfy $\m{A}\m{A}'=\m{I}$ with $'$ denoting the conjugate transpose. The soft thresholding operator $S_{\epsilon}\sbra{\m{w}}$, for a complex vector $\m{w}$, is defined as in the real case. It is also the optimal solution to (\ref{formula_soft_threshold}) for complex valued vectors $\m{w},\m{v}$. Let $\Re$ and $\Im$ be operators taking, respectively, the real and imaginary parts of a variable. The augmented Lagrangian function in \cite{yang2010orthonormal} is modified, for the CVCS, as
\equ{\cL(\m{x},\m{d},\m{y},\mu)=\onen{\m{x}}+\Re\inp{\m{d}-\Phi \m{x},\m{y}}+\frac{\mu}{2}\twon{\m{d}-\Phi \m{x}}^2,\notag}
where $\m{y}\in\bC^N$ is the Lagrange multiplier vector, $\mu\in\bR_+$ and $\inp{\m{u}_1,\m{u}_2}=\m{u}_1' \m{u}_2\in\bC$ is the inner product of $\m{u}_1,\textrm{ }\m{u}_2\in\bC^N$. The modified augmented Lagrangian function allows
a straightforward extension of the derivation and optimization steps in \cite{yang2010orthonormal} to the complex BP problem, yielding the optimal solution in the same form as that of the real valued BP. As a result, the ONE-L1 algorithms can be directly applicable to the CVCS problem. Readers are referred to \cite{yang2010orthonormal} for detailed algorithm steps and the optimality and convergence analysis.


\section{Phase Transition of CVCS}\label{section_pt_complex}

\subsection{ONE-L1 Estimation of Phase Transition of CVCS}
Section \ref{section_ONE_L1} has shown that the ONE-L1 algorithms can be extended and applicable to the CVCS problem. The eONE-L1 achieves the optimal solution of BP and rONE-L1 is numerically optimal and exponentially converges, which are applied in this subsection to empirically explore the sparsity-undersampling tradeoff of BP. The implementations of ONE-L1 algorithms follow the same procedure as their real versions in \cite{yang2010orthonormal}. We fix $r>1$ and let $\mu_{t+1}=r\cdot\mu_t$. The regulation parameter $r$ is set to $r=1+\delta$ in eONE-L1 and $r=\min\sbra{1+0.04\delta,1.02}$ is chosen in rONE-L1. The success of recovering the original signal is stated if the relative root mean squared error (RRMSE) $\twon{\hat{\m{x}}-\m{x}^o}/\twon{\m{x}^o}<10^{-4}$, where $\m{x}^o$ and $\hat{\m{x}}$ are the original and recovered signals, respectively.
Meanwhile, the failure in solving BP using ONE-L1 is stated if $\onen{\hat{\m{x}}}\geq\sbra{1+10^{-5}}\onen{\m{x}^o}$ and $\twon{\hat{\m{x}}-\m{x}^o}/\twon{\m{x}^o}\geq10^{-4}$.


Inspired by the estimation of the real phase transition in \cite{donoho2009observed,donoho2009message,yang2010orthonormal}, we first set a complex matrix ensemble, e.g., Gaussian, and dimension $N$. A grid of $\sbra{\delta,\rho}$ is generated in the plane $\mbra{0,1}\times\mbra{0,1}$ with equispaced $\delta\in\lbra{0.02,0.05,\cdots,0.98}$ and $\rho\in\lbra{\rho^R\sbra{\delta}+0.01(i-21):i=1,2,\cdots,41}$ with respect to $\delta$, where $\rho^R\sbra{\delta}$ denotes the theoretical real phase transition as shown in Fig. \ref{fig_pt_Fourier}. For each combination of $\sbra{\delta,\rho}$, $M=20$ random problem instances are generated and solved with $n=\lceil\delta N\rceil$ and $k=\lceil\rho n\rceil$. The number of success among $M$ instances is recorded. After data acquisition, a generalized linear modal (GLM) with a logistic link is used to estimate the phase transition.

We now explore the sparsity-undersampling tradeoff of BP with partial-Fourier sampling. Four values of signal length $N$ are considered, including 1024, 2048, 4096 and 8192. When $N=1024$, both eONE-L1 and rONE-L1 are used to estimate the phase transition of BP. The rONE-L1 algorithm is used for other $N$. Few failures in solving the BP occur when using rONE-L1. Fig. \ref{fig_pt_Fourier} presents the estimated phase transitions of partial-Fourier sampling. The five observed phase transitions of BP, estimated respectively by eONE-L1 with $N=1024$ and rONE-L1 with $N=1024$, 2048, 4096 and 8192, coincide with each other and are higher than the real phase transition of BP. Our earlier observation of the successful signal recovery rate reported in \cite{yang2011sparsity} also showed that a larger $N$ results in sharper phase transition. Hence, we can state:


\begin{figure}
  \centering
  \includegraphics[height=2.6in, width=3.5in]{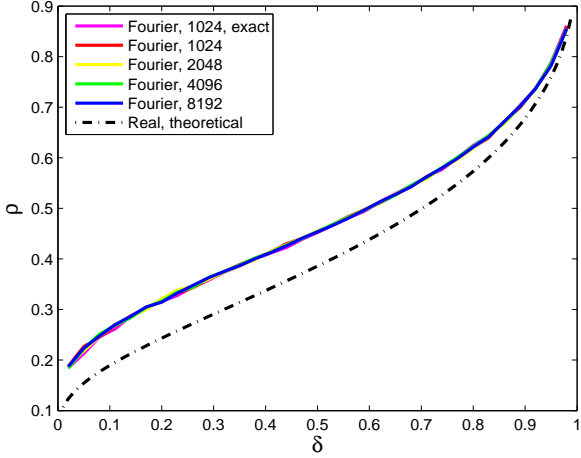}
  \caption{Observed phase transitions of partial-Fourier sampling. The upper five curves are observed phase transitions estimated by eONE-L1 with $N=1024$ and rONE-L1 with $N=1024$, $2048$, $4096$ and $8192$, respectively. The lower is the theoretical real phase transition of BP.}
  \label{fig_pt_Fourier}
\end{figure}

\begin{find}For complex signals and the partial-Fourier matrix ensemble with large dimension $N$, we observe that
\begin{itemize}
 \item[I.] BP exhibits phase transition in the plane of $\sbra{\delta,\rho}$, and a larger $N$ can result in a sharper phase transition.
 \item[II.] the complex phase transition of BP is higher than the real phase transition with a considerably enlarged success phase.
\end{itemize}\label{find_superiority}
\end{find}

\subsection{Universality of Phase Transitions of CVCS}\label{section_PT_universality}
The observed universality of phase transitions of BP in the real case has been discussed in \cite{donoho2009observed} and the same result is stated in \cite{donoho2009message,yang2010orthonormal}. In this subsection, we examine whether the same property holds in the complex case. Apart from the partial-Fourier matrix ensemble, three other complex matrix ensembles are considered with signal length $N=1000$, including Gaussian, Bernoulli and Ternary.
All random matrices have i.i.d. real and imaginary parts following the corresponding distributions. Bernoulli refers to equally likely being $0$ or $1$, and Ternary is equally likely to be $-1$, $0$ or $1$. It is noted that a matrix generated from these matrix ensembles may not be partially orthonormal as required by the ONE-L1 algorithms. This problem can always be resolved by left multiplying both sides of $\m{A}\m{x}=\m{b}$ with an invertible matrix, e.g. using QR decomposition, so the transformed equation meets the partially orthonormal condition and preserves the same solution space.

Few failures in solving BP occur in our experiment. Fig. \ref{fig_pt_universality} presents the observed phase transitions of BP with different matrix ensembles. Like the universality of phase transitions in the real case, we have the following finding.
\begin{find}
For complex signals and a number of complex matrix ensembles with large dimension $N$, BP exhibits the same phase transition.
\end{find}

\begin{figure}
  \centering
  \includegraphics[height=2.6in, width=3.5in]{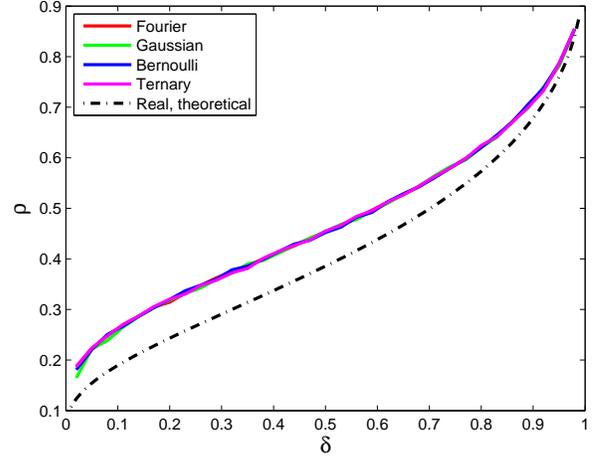}
  \caption{Observed universality of phase transitions of BP in the complex domain. The signal length $N=8192$ for partial-Fourier matrix ensemble and $N=1000$ for the other three matrix ensembles.}
  \label{fig_pt_universality}
\end{figure}

\subsection{Connection of CVCS and BSCS}
The complex system $\m{A}\m{x}=\m{b}$ can be rewritten into a real valued formulation as $\m{A}_r\m{x}_r=\m{b}_r$ with $\m{A}_r=\begin{bmatrix}\Re{\m{A}} & -\Im{\m{A}}\\\Im{\m{A}} & \Re{\m{A}}\end{bmatrix}$, $\m{x}_r=\begin{bmatrix}\Re{\m{x}}\\ \Im{\m{x}}\end{bmatrix}$ and $\m{b}_r=\begin{bmatrix}\Re{\m{b}}\\\Im{\m{b}}\end{bmatrix}$. By the $k$-sparsity of $\m{x}$, $\Re{\m{x}}$ and $\Im{\m{x}}$ are jointly $k$-sparse, in the sense that they are both $k$-sparse and share the same support. On the other hand, after proper permutations of entries of $\m{x}_r$ as well as the corresponding columns of $\m{A}_r$, $\m{A}_r\m{x}_r=\m{b}_r$ can be recast into a BSCS problem \cite{stojnic2009block} with block size $2$. The blocked signal of $\m{x}_r$ is $k$-block-sparse with at most $k$ nonzero block entries and its $\ell_{2,1}$-norm as defined in \cite{stojnic2009block} is equivalent to the $\ell_1$ norm of $\m{x}$ in the CVCS. Thus the CVCS problem is strongly connected to the BSCS problem via the real valued reformation $\m{A}_r\m{x}_r=\m{b}_r$. Their only difference is that $\m{A}_r$ is subject to a structured constraint whereas entries of the sensing matrix of the BSCS problem are independent in general.

A comparison of the phase transition of CVCS with theoretical phase transition of BSCS with block size $2$ in \cite{stojnic2009block} is presented in Fig. \ref{fig_pt_block2}. It is shown that the observed phase transition of CVCS coincides with that of BSCS.

In \cite{stojnic2009block}, the theoretical phase transition of BSCS is derived under the condition that the null space of the real valued random sensing matrix is distributed uniformly in the Grassmanian with respect to the Haar measure. It is known that the real Gaussian matrix ensemble satisfies such a condition \cite{stojnic2009reconstruction}. We now provide an analysis in the following proposition to show that such a condition is not satisfied in the CVCS problem. It therefore clarifies that the phase transition of the CVCS studied in this letter is not a special case of the phase transition result of BSCS in \cite{stojnic2009block}.

\begin{figure}
  \centering
  \includegraphics[height=2.6in, width=3.5in]{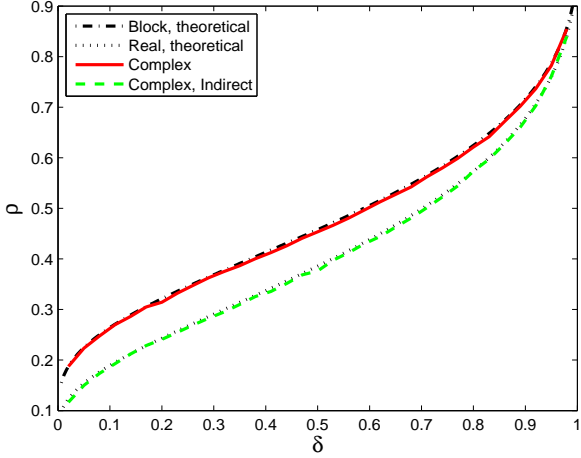}
  \caption{Coincidence between our observed complex phase transition and the block-sparse phase transition with block size of $2$. The indirect phase transition curve refers to minimizing $\onen{\m{x}_r}$ subject to $\m{A}\m{x}=\m{b}$.}
  \label{fig_pt_block2}
\end{figure}

\begin{prop} For any complex random matrix $\m{A}$ associated with the CVCS problem, the null space of $\m{A}_r$ is not distributed uniformly in the Grassmanian with respect to the Haar measure. \label{cor:nullspace}
\end{prop}
\begin{proof} Let $\cN\sbra{\m{A}_r}$ denote the null space of $\m{A}_r$. Suppose $\begin{bmatrix}\m{w}\\\m{v}\end{bmatrix}\in\cN\sbra{\m{A}_r}$ with $\m{w},\m{v}\in\bR^{N}$.
It is obvious that $\begin{bmatrix}-\m{v}\\\m{w}\end{bmatrix}\in\cN\sbra{\m{A}_r}$ and thus it is always possible to choose a basis for $\cN\sbra{\m{A}_r}$ in the form $\begin{bmatrix}\m{W} & -\m{V}\\\m{V} & \m{W}\end{bmatrix}$ with $\m{W},\m{V}\in\bR^{N\times\sbra{N-n}}$ (columns of $\m{W}+j\m{V}\in\bC^{N\times\sbra{N-n}}$ in fact compose a basis for the null space of $\m{A}$). As a result, $\cN\sbra{\m{A}_r}$ is not distributed uniformly in the Grassmanian.
\end{proof}

Proposition \ref{cor:nullspace} shows that the existing phase transition theory of BSCS cannot explain our obtained empirical phase transition of the CVCS. Intuitively, the coincidence of the phase transition of CVCS with that of BSCS may be interpreted as the universality of phase transitions of BSCS across different matrix ensembles. Fig. \ref{fig_pt_block2} also shows that the observed phase transition by minimizing $\onen{\m{x}_r}$, subject to $\m{A}\m{x}=\m{b}$ and without taking into account the block-sparsity, matches the real phase transition of BP. It further explains that the higher successful rate of CVCS is obtained by incorporating the block-sparsity factor into its problem solving.

In the above, it has been shown that the CVCS problem is connected to BSCS with a special structure of the sensing matrix $\m{A}_r$. It is interesting to consider that the sensing matrix in CVCS is further specialized into a block-diagonal matrix, i.e., $\Im\m{A}=\m{0}$ in $\m{A}_r$ or, equivalently, $\m{A}\in \bR^{n\times N}$. Such a case refers to the jointly sparse CS or the case of multiple measurement vectors (MMV). In comparison with the standard real CS or the case of single measurement vector (SMV), the recovery performance can be improved in the MMV case under some assumptions on the distribution of entries of $\m{x}$\cite{eldar2010average}. On the other hand, the analysis of the null space property in \cite{foucart2010real} shows that there is
little performance improvement in the worst case of signal recovery, such as an $\m{x}$ with identical real and imaginary parts. So, in this special case of jointly sparse CS, the observed universality of phase transitions of BSCS does not hold.

\section{Conclusion}\label{section_conclusion}
In this letter, the sparsity-undersampling tradeoff of BP for the CVCS problem is empirically explored in terms of the phase transition. It is found that the same properties exist as those in the RVCS case, such as the existence of the phase transition and observed universality across different matrix ensembles. The empirical phase transition of CVCS presents larger successful phase than that of RVCS, indicating that more successful recoveries of complex signals can be achieved by incorporating the signal structure into the problem solving process. It is shown that the empirical phase transition of CVCS is connected to and coincides with that of BSCS with block size $2$. It is however not a special case of the existing phase transition theory of BSCS and can be intuitively interpreted as an extension of current results based on the universality argument. So far, the rigorous relationship between the phase transitions of CVCS and BSCS problems is still an open problem. It is finally noted that an analysis, based on a different algorithm, of the existence of the phase transition of CVCS and its expression that agrees with the results in this letter has been reported in \cite{maleki2011asymptotic} during the preparation of this letter.

\bibliographystyle{IEEEtran}

\begin{thebibliography}{10}
\providecommand{\url}[1]{#1}
\csname url@samestyle\endcsname
\providecommand{\newblock}{\relax}
\providecommand{\bibinfo}[2]{#2}
\providecommand{\BIBentrySTDinterwordspacing}{\spaceskip=0pt\relax}
\providecommand{\BIBentryALTinterwordstretchfactor}{4}
\providecommand{\BIBentryALTinterwordspacing}{\spaceskip=\fontdimen2\font plus
\BIBentryALTinterwordstretchfactor\fontdimen3\font minus
  \fontdimen4\font\relax}
\providecommand{\BIBforeignlanguage}[2]{{%
\expandafter\ifx\csname l@#1\endcsname\relax
\typeout{** WARNING: IEEEtran.bst: No hyphenation pattern has been}%
\typeout{** loaded for the language `#1'. Using the pattern for}%
\typeout{** the default language instead.}%
\else
\language=\csname l@#1\endcsname
\fi
#2}}
\providecommand{\BIBdecl}{\relax}
\BIBdecl

\bibitem{candes2007sparsity}
E.~Cand{\`e}s and J.~Romberg, ``{Sparsity and incoherence in compressive
  sampling},'' \emph{Inverse Problems}, vol.~23, pp. 969--985, 2007.

\bibitem{candes2006compressive}
E.~Cand{\`e}s, ``Compressive sampling,'' in \emph{Proceedings of the
  International Congress of Mathematicians}, vol.~3.\hskip 1em plus 0.5em minus
  0.4em\relax Citeseer, 2006, pp. 1433--1452.

\bibitem{donoho2010counting}
D.~Donoho and J.~Tanner, ``{Counting the faces of randomly-projected hypercubes
  and orthants, with applications},'' \emph{Discrete and Computational
  Geometry}, vol.~43, no.~3, pp. 522--541, 2010.

\bibitem{blanchard2010compressed}
J.~Blanchard, C.~Cartis, and J.~Tanner, ``Compressed sensing: How sharp is the
  restricted isometry property,'' \emph{Arxiv preprint, available online at
  http://arxiv.org/abs/1004.5026}, 2010.

\bibitem{stojnic2009various}
M.~Stojnic, ``{Various thresholds for $\ell_1$-optimization in compressed
  sensing},'' \emph{Arxiv preprint, available online at
  http://arxiv.org/abs/0907.3666}, 2009.

\bibitem{donoho2009message}
D.~Donoho, A.~Maleki, and A.~Montanari, ``{Message-passing algorithms for
  compressed sensing},'' \emph{Proceedings of the National Academy of
  Sciences}, vol. 106, no.~45, pp. 18\,914--18\,919, 2009.

\bibitem{yang2010orthonormal}
Z.~Yang, C.~Zhang, J.~Deng, and W.~Lu, ``{Orthonormal expansion
  $\ell_1$-minimization algorithms for compressed sensing},'' to appear \emph{IEEE
  Transactions on Signal Processing}, 2011.

\bibitem{donoho2009observed}
D.~Donoho and J.~Tanner, ``{Observed universality of phase transitions in
  high-dimensional geometry, with implications for modern data analysis and
  signal processing},'' \emph{Philosophical Transactions of the Royal Society
  A}, vol. 367, no. 1906, pp. 4273--4293, 2009.

\bibitem{yang2011sparsity}
Z.~Yang and C.~Zhang, ``Sparsity-undersampling tradeoff of compressed sensing
  in the complex domain,'' in \emph{Acoustics, Speech and Signal Processing
  (ICASSP), IEEE International Conference on}.\hskip 1em plus 0.5em minus
  0.4em\relax IEEE, 2011, pp. 3668--3671.

\bibitem{lustig2007sparse}
M.~Lustig, D.~Donoho, and J.~Pauly, ``{Sparse MRI: The application of
  compressed sensing for rapid MR imaging},'' \emph{Magnetic Resonance in
  Medicine}, vol.~58, no.~6, pp. 1182--1195, 2007.

\bibitem{herman2009high}
M.~Herman and T.~Strohmer, ``High-resolution radar via compressed sensing,''
  \emph{IEEE Transactions on Signal Processing}, vol.~57, no.~6, pp.
  2275--2284, 2009.

\bibitem{yang2011off}
Z.~Yang, L.~Xie, and C.~Zhang, ``Off-grid direction of arrival estimation using
  sparse bayesian inference,'' \emph{Arxiv preprint, available online at
  http://arxiv.org/abs/1108.5838}, 2011.

\bibitem{stojnic2009block}
M.~Stojnic, ``Block-length dependent thresholds in block-sparse compressed
  sensing,'' \emph{Arxiv preprint, available online at
  http://arxiv.org/abs/0907.3679}, 2009.

\bibitem{stojnic2009reconstruction}
M.~Stojnic, F.~Parvaresh, and B.~Hassibi, ``On the reconstruction of
  block-sparse signals with an optimal number of measurements,'' \emph{IEEE
  Transactions on Signal Processing}, vol.~57, no.~8, pp. 3075--3085, 2009.

\bibitem{eldar2010average}
Y.~Eldar and H.~Rauhut, ``Average case analysis of multichannel sparse recovery
  using convex relaxation,'' \emph{IEEE Transactions on Information Theory},
  vol.~56, no.~1, pp. 505--519, 2010.

\bibitem{foucart2010real}
S.~Foucart and R.~Gribonval, ``Real versus complex null space properties for
  sparse vector recovery,'' \emph{Comptes rendus. Mathematique}, vol. 348, no.
  15-16, pp. 863--865, 2010.

\bibitem{maleki2011asymptotic}
A.~Maleki, L.~Anitori, Z.~Yang, and R.~Baraniuk, ``Asymptotic analysis of
  complex lasso via complex approximate message passing (camp),'' \emph{Arxiv
  preprint, available online at http://arxiv.org/abs/1108.0477}, 2011.

\end{thebibliography}


\end{document}